\documentclass[11pt,letterpaper,onecolumn]{article}

\usepackage{times}
\usepackage{graphicx}                         
\usepackage{latexsym}
\usepackage{algorithm}
\usepackage{algorithmic}
\usepackage{epsfig}
\usepackage{theorem}
\usepackage{dsfont}
\usepackage{pstricks,pst-node,pst-plot,fancyheadings}
\usepackage{amsmath}
\usepackage{amssymb}
\usepackage{fullpage}
\usepackage{multirow}

\newcommand{\N}{\mathds{N}}

\newtheorem{definition}{Definition}
\newtheorem{theorem}{Theorem}
\newtheorem{corollary}{Corollary}
\newtheorem{lemma}{Lemma}
\newenvironment{proof}{\noindent \begin{rm}{\textbf{Proof.} }}{\hspace*{\fill}$\Box$\par\end{rm} \vspace{.3cm}}
\newtheorem{spec}{Specification}

\newtheorem{rem}{Remark}

\newcommand{\id}{\mbox{ID}}

\begin{document}


\title{Self-stabilizing algorithms for Connected Vertex Cover and Clique decomposition problems}

\author{Fran\c{c}ois Delbot\\
Universit\'{e} Paris Ouest Nanterre / LIP6, CNRS UMR 7606.\\4 place Jussieu, 75252 Paris Cedex, France.\\ \texttt{francois.delbot@lip6.fr}\\
\and Christian Laforest\\
Universit\'{e} Blaise Pascal / LIMOS, CNRS UMR 6158, ISIMA.\\Campus scientifique des C\'{e}zeaux, 24 avenue des Landais, 63173 Aubiere cedex, France.\\ \texttt{christian.laforest@isima.fr}\\
\and Stephane Rovedakis\\
Conservatoire National des Arts et M\'{e}tiers / CEDRIC, EA 4629.\\292 rue Saint-Martin, F-75141 Paris Cedex 03, France.\\ \texttt{stephane.rovedakis@cnam.fr}
}

\maketitle

\begin{abstract}
In many wireless networks, there is no fixed physical backbone nor centralized network management. The nodes of such a network have to self-organize in order to maintain a virtual backbone used to route messages. Moreover, any node of the network can be \textit{a priori} at the origin of a malicious attack. Thus, in one hand the backbone must be fault-tolerant and in other hand it can be useful to monitor all network communications to identify an attack as soon as possible. We are interested in the minimum \emph{Connected Vertex Cover} problem, a generalization of the classical minimum Vertex Cover problem, which allows to obtain a connected backbone. Recently, Delbot \emph{et al.}~\cite{DelbotLP13} proposed a new centralized algorithm with a constant approximation ratio of $2$ for this problem. In this paper, we propose a distributed and self-stabilizing version of their algorithm with the same approximation guarantee. To the best knowledge of the authors, it is the first distributed and fault-tolerant algorithm for this problem. The approach followed to solve the considered problem is based on the construction of a connected minimal clique partition. Therefore, we also design the first distributed self-stabilizing algorithm for this problem, which is of independent interest.\\

\end{abstract}


\section{Introduction}

In many wireless networks, there is no fixed physical backbone nor centralized network management. In such networks, the nodes need to regularly flood control messages which leads to the "broadcast  storm  problem"~\cite{NiTCS99}. Thus, the nodes have to self-organize in order to maintain a virtual backbone, used to route messages in the network. Routing messages are only exchanged inside the backbone, instead of being broadcasted to the entire network. To this end, the backbone must be connected.
The construction and the maintenance of a virtual backbone is often realized by constructing a Connected Dominating Set. A \emph{Connected Dominating Set (CDS)} of a graph $G=(V,E)$ is a set of nodes $S\subseteq V$ such that $G[S]$ (the graph induced by $S$ in $G$) is connected and each node in $V-S$ has at least one neighbor in $S$. Nodes from $S$ are responsible of routing the messages in the network, whereas nodes in $V-S$ communicate by exchanging messages through neighbors in $S$. In order to minimize the use of resources, the size of the backbone (and thus of the  CDS) is minimized. This problem is NP-hard~\cite{GJ79} and has been extensively studied due to its importance for communications in wireless networks. Many algorithms have been proposed in centralized systems (e.g., see~\cite{BlumDTC99} for a survey). In addition to message routing, there is the problem of network security. Indeed, a faulty node infected by a virus or an unscrupulous user can be at the origin of flooding or a malicious attack. Thus, it is necessary to monitor all network communications to identify these situations, as soon as possible, in order to isolate this node. A CDS $S$ will not support this feature since two nodes in $V-S$ can be neighbors, i.e, $V-S$ is not always an independent set.

In order to monitor all network communications, we can consider the Vertex Cover problem. A \emph{vertex cover} of a graph $G=(V,E)$ is a set of nodes $S \subseteq V$ such that each edge $e=uv$ is \emph{covered} by $S$, i.e., $u \in S$ or $v \in S$ (or both). A vertex cover is \emph{optimal} if it's size is minimum. This is a classical NP-complete problem~\cite{GJ79} that can be approximated with a ratio of $2$. However, if a vertex cover allows to monitor all network communications, it is not always connected and cannot be used as a backbone. A \emph{Connected vertex cover} $S$ of $G$ is a vertex cover of $G$ with the additional property that $G[S]$ (the graph induced by $S$ in $G$) is connected. Similarly, an \emph{optimal} connected vertex cover is one of minimum size and the associated problem is also NP-complete. Not a lot of work has been done on this problem (see~\cite{Savage82,EscoffierGouvesMonnot10}). More recently, Delbot \emph{et al.} in~\cite{DelbotLP13} proposed another (centralized) $2$-approximation algorithm based on connected clique partitions of $G$.

In practice, it is more convenient to use distributed and fault-tolerant algorithms, instead of centralized algorithms due to the communications cost to obtain the network topology. \emph{Self-stabilization} introduced first by Dijkstra in~\cite{Dijkstra74,Dolev2000} is one of the most versatile techniques to ensure a distributed system to recover a correct behaviour. A distributed algorithm is self-stabilizing if after faults and attacks hit the system and place it in some arbitrary global state, the system recovers from this catastrophic situation without external (e.g., human) intervention in finite time. Many self-stabilizing algorithms have been proposed to solve a lot of graph optimization problems, e.g., Guellati and Kheddouci~\cite{GuellatiK10} give a survey for several problems related to independence, domination, coloring and matching in graphs. For the minimum CDS problem, Jain and Gupta~\cite{JainG05} design the first self-stabilizing algorithm for this problem. More recently, Kamei \emph{et al.}~\cite{KameiK10,KameiK12,KameiIY13} proposed several self-stabilizing algorithms with a constant approximation ratio and an additional property during the algorithm convergence.\\
However, as explained above a CDS does not meet all the desired properties. This is why we study the minimum connected vertex cover from a distributed and self-stabilizing point of view.

\paragraph{Contributions.}
We consider the minimum \emph{Connected Vertex Cover} problem in a distributed system subject to transient faults. In this paper, we propose a distributed and self-stabilizing version of the algorithm given recently by Delbot \emph{et al.}~\cite{DelbotLP13} for this problem while guaranteeing the same approximation ratio of 2. To the best of our knowledge, it is the first distributed and fault-tolerant algorithm for this problem. The approach followed to solve the considered problem is based on the construction of a \emph{Connected Minimal Clique Partition}. Therefore, we also design the first distributed self-stabilizing algorithm for this problem, which is of independent interest. Moreover, these algorithms works under the distributed daemon without any fairness assumptions (which is the weakest daemon).\\

The rest of this paper is organized as follows. The next section describes the model considered in the paper and the notations used. In Section~\ref{sec:cmcp}, we consider first the Connected Minimal Clique Partition problem. We give a state of the art related to the graph decomposition problem, then we present our self-stabilizing algorithm for this problem and prove its correctness. Section~\ref{sec:cvc} is devoted to the Connected Vertex Cover problem. We introduce first related works associated with this problem, then we give the self-stabilizing connected vertex cover algorithm that we propose and we give the correctness proof. Finally, the last section concludes the paper and present several perspectives.

\section{Model}
\label{sec:model}

\paragraph{Notations.}
We consider a network as an undirected connected graph $G=(V,E)$ where $V$ is a set of nodes (or {\em processors}) and $E$ is the set of {\em bidirectional asynchronous communication links}. We state that $n$ is the size of $G$ ($|V| = n$) and $m$ is the number of edges ($|E| = m$). We assume that the graph $G=(V,E)$ is a simple connected graph. In the network, $p$ and $q$ are neighbors if and only if a communication link ($p$,$q$) exists (i.e., ($p$,$q$) $\in$ $E$). Every processor $p$ can distinguish all its links. To simplify the presentation, we refer to a link ($p$,$q$) of a processor $p$ by the {\em label} $q$. We assume that the labels of $p$, stored in the set $Neig_p$, are locally ordered by $\prec_p$. We also assume that $Neig_p$ is a constant input from the system. $Diam$ and $\Delta$ are respectively the diameter and the maximum degree of the network (i.e., the maximal value among the local degrees of the processors). Each processor $p \in V$ has a unique identifier in the network, noted $\id_p$.

\paragraph{Programs.}In our model, protocols are {\em uniform}, i.e., each processor executes the same program. We consider the local shared memory model of computation. In this model, the program of every processor consists in a set of {\em variables} and an {\em ordered finite set of actions} inducing a {\em priority}. This priority follows the order of appearance of the actions into the text of the protocol. A processor can write to its own variable only, and read its own variables and that of its neighbors. Each action is constituted as follows: $<label>\ ::\ <guard>\ \to \ <statement>.$ The guard of an action in the program of $p$ is a boolean expression involving variables of $p$ and its neighbors. The statement of an action of $p$ updates one or more variables of $p$. An action can be executed only if its guard is satisfied. The {\em state} of a processor is defined by the value of its variables. The {\em state} of a system is the product of the states of all processors. We will refer to the state of 
a processor and the system as a ({\em local}) {\em state} and ({\em global}) {\em configuration}, respectively. We note $\mathcal C$ the set of all possible configuration of the system. Let $\gamma \in \mathcal C$ and $A$ an action of $p$ ($p \in V$). $A$ is said to be {\em enabled} at $p$ in $\gamma$ if and only if the guard of $A$ is satisfied by $p$ in $\gamma$. Processor $p$ is said to be {\em enabled} in $\gamma$ if and only if at least one action is enabled at $p$ in $\gamma$. When several actions are enabled simultaneously at a processor $p$: only the priority enabled action can be activated.

Let a distributed protocol $P$ be a collection of binary transition relations denoted by $\mapsto$, on $\mathcal C$. A $\emph{computation}$ of a protocol $P$ is a $\emph{maximal}$ sequence of configurations $e=(\gamma _{0}$,$\gamma_{1}$,...,$\gamma_{i}$,$\gamma _{i+1}$,...$)$ such that, $\forall i\geq 0$, $\gamma_{i}\mapsto \gamma _{i+1}$ (called a {\em step}) if $\gamma_{i+1}$ exists, else $\gamma_{i}$ is a terminal configuration. $\emph{Maximality}$ means that the sequence is either finite (and no action of $P$ is enabled in the terminal configuration) or infinite. All computations considered here are assumed to be maximal. $\mathcal E$ is the set of all possible computations of $P$.

As we already said, each execution is decomposed into steps. Each step is shared into three sequential phases atomically executed: $(i)$ every processor evaluates its guards, $(ii)$ a {\em daemon} (also called {\em scheduler}) chooses some enabled processors, $(iii)$ each chosen processor executes its priority enabled action. When the three phases are done, the next step begins. 

A {\em daemon} can be defined in terms of {\em fairness} and {\em distributivity}. In this paper, we use the notion of {\em unfairness}: the {\em unfair} daemon can forever prevent a processor from executing an action except if it is the only enabled processor. Concerning the {\em distributivity}, we assume that the daemon is {\em distributed} meaning that, at each step, if one or more processors are enabled, then the daemon chooses at least one of these processors to execute an action.

We consider that any processor $p$ executed a {\em disabling} {\em action} in the computation step $\gamma_{i}\mapsto \gamma _{i+1}$ if $p$ was {\em enabled} in $\gamma_{i}$ and not enabled in $\gamma_{i+1}$, but did not execute any protocol action in $\gamma_{i}\mapsto \gamma _{i+1}$. The disabling action represents the following situation: at least one neighbor of $p$ changes its state in $\gamma_i \mapsto \gamma_{i+1}$, and this change effectively made the guard of all actions of $p$ false in $\gamma_{i+1}$.

To compute the time complexity, we use the definition of (asynchronous) \emph{round}. This definition captures the execution rate of the slowest processor in any computation. Given a computation $e$ ($e \in \mathcal{E}$), the \emph{first round} of $e$ (let us call it $e^{\prime}$) is the minimal prefix of $e$ containing the execution of one action (an action of the protocol or a disabling action) of every enabled processor from the initial configuration.  Let $e^{\prime \prime}$ be the suffix of $e$ such that $e=e^{\prime}e^{\prime \prime}$. The \emph{second round} of $e$ is the first round of $e^{\prime \prime}$, and so on.

\section{Connected Minimal Clique Partition problem}
\label{sec:cmcp}

In this section, we consider a first problem whose aim is the partitioning of the input graph into subgraphs of maximal size in a distributed fashion, while maintaining a connectivity constraint between some subgraphs. More particularly, the goal is to decompose an input undirected graph $G=(V,E)$ into a set of cliques of maximal size such that all cliques of size at least two are connected. The connectivity constraint can be used for communication facilities. In the following, we define more formally the Connected Minimal Clique Partition problem.

\begin{definition}[Connected Minimal Clique Partition]
\label{def:cmcp}
Let $G=(V,E)$ be any undirected graph, and a \emph{clique} is a complete subgraph of $G$. A clique partition $C_1, \dots, C_k$ of $G$ is \emph{minimal} if for all $i\neq j$ the graph induced by $C_i \cup C_j$ is not a clique. A minimal clique partition $C_1, \dots, C_k$ is \emph{connected} iff for any pair of nodes $u,v$ in $\bigcup_{1 \leq i \leq l} C_i$, with $C_i$ the non trivial cliques of the partition and $l \leq k$, there is a path between $u$ and $v$ in the graph induced by $\bigcup_{1 \leq i \leq l} C_i$.
\end{definition}

Since we consider that faults can arise in the system, we give in Specification~\ref{spec:cmcp} the conditions that a self-stabilizing algorithm solving the Connected Minimal Clique partition problem have to satisfy.

\begin{spec}[Self-stabilizing Connected Minimal Clique Partition]
\label{spec:cmcp}
Let $\mathcal{C}$ be the set of all possible configurations of the system. An algorithm $\mathcal{A_{CMCP}}$ solving the problem of constructing a stabilizing connected minimal clique partition satisfies the following conditions:
\begin{enumerate}
\item Algorithm $\mathcal{A}$ reaches a set of terminal configurations $\mathcal{T} \subseteq \mathcal{C}$ in finite time, and
\item Every configuration $\gamma \in \mathcal{T}$ satisfies Definition~\ref{def:cmcp}.
\end{enumerate}
\end{spec}

\subsection{Related works}

The decomposition of an input graph into patterns or partitions has been extensively studied in the literature, and also in the self-stabilizing context. Most of graph partitioning problems are NP-complete. For the graph decomposition into patterns, Ishii and Kakugawa~\cite{IshiiK02} proposed a self-stabilizing algorithm for the construction of cliques in a connected graph with unique nodes identifier. Each process has to compute the largest set of cliques of same maximum size it can belong to in the graph. A set of cliques is constructed in $O(n^4)$ computation steps assuming an unfair centralized daemon. Moreover, the authors show that there exists no self-stabilizing algorithm in arbitrary anonymous graphs for this problem. Neggazi \emph{et al.}~\cite{NeggaziHK12} considered the problem of decomposing a graph into a maximal set of disjoint triangles. They give the first self-stabilizing algorithm for this problem whose convergence time is $O(n^4)$ steps under an unfair central daemon with unique nodes identifier. Neggazi \emph{et al.}~\cite{NeggaziTHK13} studied later the uniform star decomposition problem, i.e., the goal is to divide the graph into a maximum set of disjoint stars of $p$ leaf nodes. This is a generalization of the maximum matching problem which is a NP-complete problem constructing a maximum set of independent edges of the graph. Thus, a 1-star decomposition is equivalent to a maximum matching. The authors proposed a self-stabilizing algorithm constructing a maximal $p$-star decomposition of the input graph in $O(\frac{n}{p+1})$ asynchronous rounds and a (exponential) bounded number of steps under an unfair distributed daemon with unique nodes identifier.\\
A well studied problem related with graph decomposition is the maximum matching problem. Many works address the maximal matching problem which is polynomial. The first self-stabilizing algorithm for this problem has been proposed by Hsu \emph{et al.}~\cite{HsuH92}. The algorithm converges in $O(n^4)$ steps under a centralized daemon. Hedetniemi \emph{et al.}~\cite{HedetniemiJS01} showed later that the algorithm proposed by Hsu \emph{et al.} has a better convergence time of $2m+n$ steps under a centralized daemon. Goddar \emph{et al.}~\cite{GoddardHJS03} considered the construction of a maximal matching in ad-hoc networks and give a solution which stabilizes in $n+1$ rounds under a synchronous distributed daemon. Manne \emph{et al.}~\cite{ManneMPT09} have shown that there exists no self-stabilizing algorithm for this problem under a synchronous distributed daemon in arbitrary anonymous networks. They proposed an elegant algorithm which converges in $O(n)$ rounds and $O(m)$ steps under an unfair distributed daemon in arbitrary networks with unique nodes identifier. Recently, several works consider the maximum matching problem to find an optimal or an approximated solution. Hadid \emph{et al.}~\cite{HadidK09} give an algorithm which constructs an optimal solution in $O(Diam)$ rounds under a weakly fair distributed daemon only in bipartite graphs. Manne \emph{et al.}~\cite{ManneMPT11} presented a self-stabilizing algorithm constructing a $\frac{2}{3}$-approximated maximum matching in general graphs within $O(n^2)$ rounds and a (exponential) bounded number of steps under an unfair distributed daemon. Manne \emph{et al.}~\cite{ManneM07} proposed the first self-stabilizing algorithm for the maximum weighted matching problem achieving an approximation ratio of 2 in a (exponential) bounded number of steps under a centralized daemon and a distributed daemon. Turau \emph{et al.}~\cite{TurauH11a} gave a new analysis of the algorithm of Manne \emph{et al.}~\cite{ManneM07}. They showed that this algorithm converges in $O(nm)$ steps under a centralized daemon and an unfair distributed daemon.\\
More recently, some self-stabilizing works investigated the graph decomposition into disjoint paths. Al-Azemi \emph{et al.}~\cite{Al-AzemiK11} studied the decomposition of the graph in two edge-disjoint paths in general graphs, while Neggazi \emph{et al.}~\cite{NeggaziHK12b} considered the problem of dividing the graph in maximal disjoint paths of length two.
Finally, the partitioning in clusters of the input graph has been extensively studied. Belkouch \emph{et al.}~\cite{BelkouchBCD02} proposed an algorithm to divide a graph of order $k^2$ into $k$ partitions of size $k$. The algorithm is based on spanning tree constructions of height $h$ and converges in $O(h)$ rounds under a weakly fair distributed daemon. Johnen \emph{et al.}~\cite{JohnenN09} studied the weighted clustering problem and introduced the notion of robustness allowing to reach quickly (after one round) a cluster partition. A cluster partition is then preserved during the convergence to a partition satisfying the clusterhead's weight. Bein \emph{et al.}~\cite{BeinDJV05} design a self-stabilizing clustering algorithm dividing the network into non-overlapping clusters of depth two, while Caron \emph{et al.}~\cite{CaronDDL10} considered the $k$-clustering problem in which each node is at most at distance $k$ from its clusterhead. Recently, Datta \emph{et al.}~\cite{DattaLDHR12} design a self-stabilizing $k$-clustering algorithm guaranteeing an approximation ratio in unit disk graphs.\\

All the works presented above concern the graph decomposition problem using different patterns. However, none of them allow to construct a disjoint maximal clique partition of the graph. Note that Ishii and Kakugawa~\cite{IshiiK02} computes a set of maximal cliques which are not necessary disjoint. Moreover, the \emph{non trivial} cliques (with at least two nodes) of the partition must be connected. 

 In~\cite{DelbotLP13}, the authors are interested to the decomposition of an input graph in cliques while satisfying a connectivity property. They propose a centralized algorithm for the Connected Minimal Clique Partition problem (see Definition~\ref{def:cmcp}). The proposed algorithm constructs iteratively a set of maximal cliques $S$. At the beginning of the algorithm, $S$ is empty and a node $u_1 \in V$ is randomly (with equiprobability) selected. A first maximal clique $C_1$ containing $u_1$ is added to $S$ and all the nodes of $C_1$ are marked in $G$. Then for any iteration $i$, any non marked node $u_i \in V, 1 \leq i \leq k,$ neighbor of at least one marked node of $G$ is randomly (with equiprobability) selected. As for the first clique, a new maximal clique containing $u_i$ is greedily built among non marked nodes of $G$. This procedure is executed 
iteratively while there is a non marked node in $G$. As mentioned in~\cite{DelbotLP13}, every \emph{trivial} clique (clique of size one) in the constructed set $S$ is neighbor of no other trivial clique. So the set of trivial cliques of any minimal partition computed by this algorithm induces an independent set of $G$. Otherwise, it could be possible to merge two trivial cliques of $S$ in order to obtain a clique of size two.

\subsection{Self-stabilizing construction}

In this subsection we present the self-stabilizing algorithm called $\mathcal{SS-CMCP}$ for the Minimal Clique Partition problem, a formal description is given in Algorithm~\ref{algo}.

\paragraph{General overview}

The self-stabilizing algorithm $\mathcal{SS-CMCP}$ is based on the approach proposed by Delbot \emph{et al.}~\cite{DelbotLP13} (see description in the precedent subsection). In order to design a distributed version of this approach, we consider here a designated node in the network called the \emph{root} node, noted $r$ in the following, and distances (in hops) from $r$ given in input at each node $p$ noted $dist_p$. These distance values can be obtained by computing a BFS tree rooted at $r$. Several self-stabilizing BFS algorithms can be used, e.g.,~\cite{HuangC92,DolevIM93,Johnen97,CournierRV11}. As described below, we use these information to define an order on the construction of the clique partition of the graph.\\
In the proposed algorithm, the construction of maximal cliques is performed starting from the root $r$ and following the distances in the graph. Indeed, the pair \emph{(distance, node identifier)} allows to define a construction priority for the cliques. First of all, each node shares the set of its neighbors with its neighborhood, allowing for each node to know its 2-hops neighborhood. The 2-hops neighborhood is used by each node to identify amongs its neighbors the ones which can belong to its maximal clique. For each node $p$, we define by \emph{candidate leaders} the set of neighbors $q$ of $p$ such that the pair $(dist_q, \id_q)$ is lexicographically smaller than $(dist_p, \id_p)$. In Algorithm $\mathcal{SS-CMCP}$, each node $p$ can construct its maximal clique by selecting in a greedily manner a set of neighbors $S \subseteq Neig_p$ such that (i) for any $q \in S$ we have $(dist_p, \id_p)<(dist_q, \id_q)$ and (ii) $S \cup \{p\}$ is a complete subgraph. This computation is performed by any node $p$ which has not been selected by one of its candidate leaders. In this case, $p$ is called a \emph{local leader}, otherwise $p$ is no more a local leader and clears out its set $S$. Each node selected by one of its candidate leaders has to accept only the selection of its candidate leader $q$ of smallest pair $(dist_q, \id_q)$. Finally, any local leader $p$ which has initiated the construction of its maximal clique considers in its clique only the selected neighbors which have accepted $p$'s selection.\\
The proposed algorithm maintains a connectivity property between non trivial cliques of the constructed partition. This is a consequence of the construction order of the maximal cliques, which follows the distances in the network from $r$. Indeed, every non trivial clique $C_i$ (that does not contain the root node $r$) is adjacent to at least another non trivial clique $C_j$, such that $dist_{l_j} \leq dist_{l_i}$ with $l_k$ the local leader of the clique $C_k$. Otherwise, by construction another local leader $l_g$, with $dist_{l_g} \leq dist_{l_i}$, selects the local leader $l_i$ to belong to its maximal clique. As a consequence, the maximal clique $C_i$ is removed. In fact, the algorithm constructs a specific clique partition among the possible partitions that the centralized approach proposed in~\cite{DelbotLP13} can compute.

\paragraph{Detailed description}

In the following, we give more details on the proposed algorithm $\mathcal{SS-CMCP}$. Our algorithm is composed of four rules executed by every node and five variables are maintained at each node:

\begin{itemize}
\item $N_p$: this variable contains the set of neighbors of $p$ which allows to each node to be informed of the 2-hops neighborhood,
\item $d_p$: this variable is used to exchange the value of $dist_p$ with $p$'s neighbors,
\item $S_p$: this variable is used by $p$ to indicate in its neighborhood the nodes selected by $p$ (if $p$ is a local leader),
\item $C_p$: this variable contains the set of nodes which belong to the maximal clique of $p$ (if $p$ is a local leader),
\item $lead_p$: this variable stores the local leader in the neighborhood of $p$.
\end{itemize}

As explained above, each node stores in variable $N_p$ the set of its 1-hop neighborhood, this is done using the first rule $N$-action of the algorithm which is executed in case we have $N_p \neq Neig_p$. The information stored in this variable is used by each node in $p$'s neighborhood for the computation of maximal cliques. For each node, the set of candidate leaders is given by Macro $LNeig_p$, and among this set of nodes the Macro $SNeig_p$ indicates the neighbors which have selected $p$ for the construction of their own maximal clique. Every node $p$ which is not selected by a candidate leader does not satisfy Predicate $Selected(p)$ and can execute $C1$-action to start the construction of its own maximal clique. The procedure $Clique\_temp()$ selects in a greedily manner the neighbors which forms with $p$ a complete subgraph. By executing $C1$-action, a node $p$ becomes a local leader by storing its identifier in its variable $lead_p$ and notifies with its variable $S_p$ the neighbors it has selected using Procedure $Clique\_temp()$. $C1$-action can be executed by a node $p$ only if $S_p$ does not contain the correct set of selected neighbors, i.e., we have $S_p \neq Clique\_temp()$. Then, each node $p$ selected by a candidate leader (i.e., which satisfies Predicate $Selected(p)$) can execute $C2$-action to accept the selection of its candidate leader $q$ of smallest pair $(dist_q, \id_q)$. In this case, we say that $q$ has been elected as the local leader of $p$. This is given by Macro $Leader_p$ and stored in the variable $lead_p$. $C2$-action is only executed if the variable $lead_p$ does not store the correct local leader for $p$, i.e., we have $lead_p \neq Leader_p$. Finally, $C3$-action allows to each local leader $p$ to establish the set of neighbors $q$ which are contained in its maximal clique. This set is stored in variable $C_p$ and is given by Macro $Clique(p)$ considering only the neighbors $q$ of $p$ which have elected $p$ as their local leader (i.e., $lead_q=\id_p$). This last rule is executed only by local leaders 
which are not selected to belong to another clique (i.e., $Selected(p)$ is not satisfied) and have not computed the correct set of neighbors contained in their maximal clique (i.e., $S_p=Clique\_temp()$ and $C_p \neq Clique_p$).

\begin{algorithm}[t]
\caption{\quad Self-Stabilizing Connected Minimal Clique Partition algorithm for any $p \in V$\label{algo}}
\smallskip
\begin{scriptsize}
{\bf Inputs:}\\
\hspace*{1cm}$Neig_p$: set of (locally) ordered neighbors of $p$;\\
\hspace*{1cm}$\id_p$: unique identifier of $p$;\\
\hspace*{1cm}$dist_p$: distance between $p$ and the root (leader node);\\
{\bf Variables:}\\
\hspace*{1cm}$N_p$: variable used to exchange the neighbor set $Neig_p$ in $p$'s neighborhood, $N_p \subseteq Neig_p$;\\
\hspace*{1cm}$d_p$: variable used to exchange the distance $dist_p$ in $p$'s neighborhood, $d_p \in \N$;\\
\hspace*{1cm}$S_p$: variable used by $p$ to select neighbors for the construction of its maximal clique, $S_p \subseteq Neig_p$;\\
\hspace*{1cm}$C_p$: variable used to store the set of neighbors belonging to the maximal clique of $p$, $C_p \subseteq Neig_p$;\\
\hspace*{1cm}$lead_p$: variable used to store the local leader of $p$, $lead_p \in Neig_p$;

.\dotfill\ 

{\bf Macros:}\\
\begin{tabular}{lll}
$Clique_p$ & $=$ & $\{q \in S_p: lead_q=\id_p\}$\\
$LNeig_p$ & $=$ & $\{q \in Neig_p: d_q<d_p \vee (d_q=d_p \wedge \id_q < \id_p)\}$\\
$SNeig_p$ & $=$ & $\{q \in LNeig_p: p \in S_q\}$\\
$Leader_p$ & $=$ & $\left\{\begin{array}{lll}\bot & & \mbox{If }SNeig_p=\emptyset\\ \min \{q \in SNeig_p: (\forall s \in SNeig_p: d_q \leq d_s)\} & & \mbox{Otherwise}\end{array}\right.$\\
\end{tabular}\\
.\dotfill\

{\bf Predicate:}\\
\begin{tabular}{lll}
$Selected(p)$ & $\equiv$ & $SNeig_p \neq \emptyset$\\
\end{tabular}\\
.\dotfill\

{\bf Procedure:}\\
Clique\_temp()
\begin{algorithmic}[1]
\STATE{$S:=\{p\};$}
\FORALL{$q \in (Neig_p-LNeig_p)$}
\IF{$S \subseteq N_q$}
\STATE{$S:=S \cup \{q\};$}
\ENDIF
\ENDFOR
\RETURN $S$;

\end{algorithmic}
.\dotfill\

{\bf Actions:}\\
\begin{tabular}{lllllll}
&& $N$-$action$  & $::$ & $N_p \neq Neig_p \vee d_p \neq dist_p$ & $\to$ & $N_p:=Neig_p; d_p:=dist_p$;\\
&& $C1$-$action$ & $::$ & $\neg Selected(p) \wedge S_p \neq Clique\_temp()$ & $\to$ & $S_p:=Clique\_temp();lead_p:=\id_p$;\\
&& $C2$-$action$ & $::$ & $Selected(p) \wedge lead_p \neq Leader_p$ & $\to$ & $lead_p:=Leader_p; S_p:=\emptyset; C_p:=\emptyset$;\\
&& $C3$-$action$ & $::$ & $\neg Selected(p) \wedge S_p=Clique\_temp() \wedge C_p \neq Clique_p$ & $\to$ & $C_p:=Clique_p$;\\
\end{tabular}
\end{scriptsize}
\end{algorithm}

\paragraph{Example of an execution}

We illustrate with an example given in Figure~\ref{fig:ex_algo1} how the proposed algorithm $\mathcal{SS-CMCP}$ constructs a Connected Minimal Clique Partition. In this example, we consider a particular execution following the distances in the graph and we give only the correct cliques which are constructed by the algorithm. We consider the topology given in Figure~\ref{fig:ex_algo1}(a). First of all, each node exchanges its neighbors set using $N$-action. The root node $r$ cannot be selected by one of its neighbors, so by executing $C1$-action it becomes a local leader (i.e., $lead_r= \id_r$) and selects among its neighbors the nodes to include in its maximal clique, i.e., by indicating in its variable $S_r$ the nodes 1, 2 and 5. Then, nodes 1, 2 and 5 detect that they have been selected by $r$ (their unique possible candidate leader) and in response they elect $r$ using $C2$-action. The node $r$ executes $C3$-action to construct its maximal clique by adding in its variable $C_r$ the nodes which have elected $r$ as their local leader, i.e., nodes 1, 2 and 5, as illustrated in Figure~\ref{fig:ex_algo1}(b). Next, the nodes 3, 4 and 6 elect themselves as local leaders since they are not selected to belong to a clique. They execute $C1$-action to select among their neighbors of equal or higher distance those which forms a complete subgraph (including themselves), i.e., neighbors 10 and 15 for node 3, neighbor 7 for node 4 and neighbor 9 for 6. The selected neighbors execute $C2$-action to elect the single candidate leader neighbor which has selected them to join a clique. We remind that in case of a selection from multiple candidate leaders a selected node elects the candidate leader $x$ of smallest pair $(dist_x, \id_x)$ with Macro $Leader$. Then, the local leaders 3, 4 and 6 execute $C3$-action to construct respectively their maximal clique as illustrated in Figure~\ref{fig:ex_algo1}(c). In the same way, nodes 8 and 12 become local leaders and select respectively no neighbor and neighbors 11 and 14 to join their clique. The neighbors selected by node 12 elect 12 as their local leader and node 12 constructs its maximal clique, while node 8 constructs a trivial clique as illustrated in Figure~\ref{fig:ex_algo1}(d). Finally, node 13 becomes a local leader and constructs a trivial clique as illustrated in Figure~\ref{fig:ex_algo1}(e), which gives the complete clique partition constructed by the algorithm.

\begin{figure}[!ht]
\begin{center}
\includegraphics[scale=0.6]{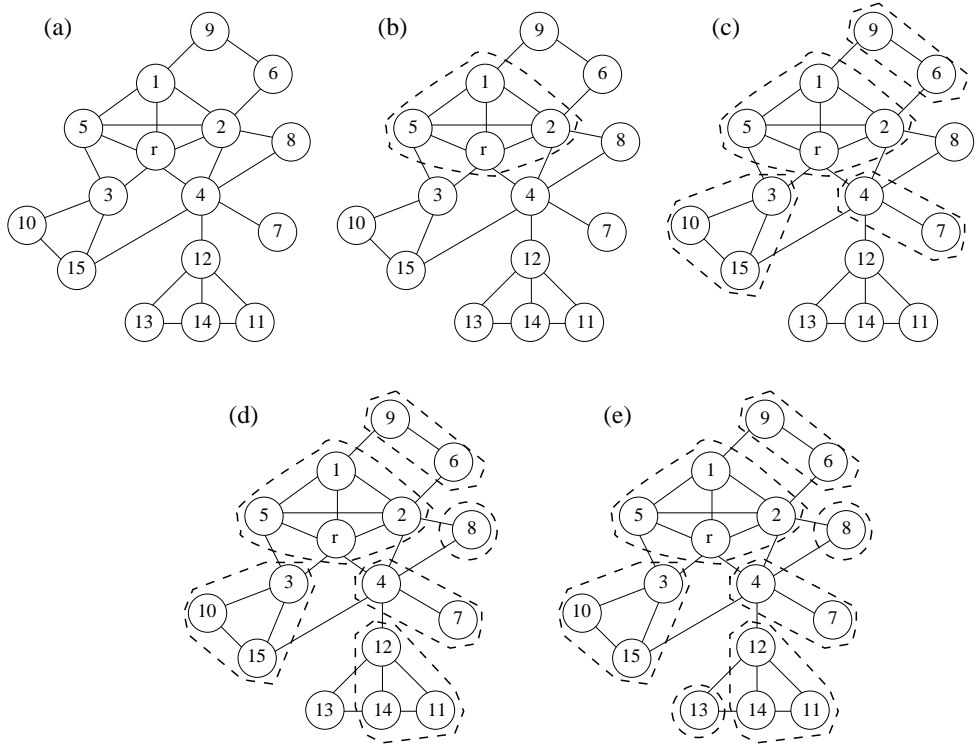}
\end{center}
\caption{Execution of Algorithm $\mathcal{SS-CMCP}$.}
\label{fig:ex_algo1}
\end{figure}

\subsection{Correctness proof}

\begin{definition}[Rank of a node]
The \emph{rank} of any node $p \in V$ is defined by the pair $(dist_p, \id_p)$. Given two nodes $p, q \in V, p \neq q$, we say that the rank of $p$ is higher than the rank of $q$, noted $rank(p) \prec rank(q)$, iff either $dist_p<dist_q$, or $dist_p=dist_q$ and $\id_p<\id_q$.
\end{definition}

\begin{definition}[Selection of nodes]
A node $q \in V$ is \emph{selected} by a neighbor $p$ of $q$ if we have $rank(p) \prec rank(q)$ and $q \in S_p$ (i.e., Predicate $Selected(q)$ is satisfied at $q$).
\end{definition}

\begin{definition}[Local leader]
\label{def:local_leader}
Given any clique partition $C_1, \dots, C_k$ of a graph $G=(V,E)$, a node $p_i \in V, 1 \leq i \leq k,$ is a \emph{local leader} of a clique $C_i$ if we have $p_i \in C_i$ and $p_i$ is not selected (i.e., we have $\neg Selected(p_i)$ at $p_i$).
\end{definition}

\begin{definition}[Rank of a clique]
Given any clique partition $C_1, \dots, C_k$ of a graph $G=(V,E)$, the rank associated to a clique $C_i, 1 \leq i \leq k,$ is equal to the rank of the local leader $p_i$ of $C_i$.
\end{definition}

\begin{rem}
\label{rem:total_order_clique}
Given any clique partition $C_1, \dots, C_k$ of a graph $G=(V,E)$, the rank of the cliques define a total order on the cliques of the partition in $G$.
\end{rem}

\begin{definition}[Election of a local leader]
Let $G=(V,E)$ be any graph and $p \in V$ a node selected by a local leader $p_i \in V$. $p$ has \emph{elected} $p_i$ to join its clique if we have $lead_p=\id_{p_i}$.
\end{definition}

\begin{definition}[Correct clique]
\label{def:correct_clique}
Given a clique partition $C_1, \dots, C_k$ of a graph $G=(V,E)$, a clique $C_i, 1 \leq i \leq k,$ is \emph{correct} iff the following conditions are satisfied:
\begin{enumerate}
\item There is a single local leader $p_i \in V$ in $C_i$;
\item $p_i$ has selected a subset $S_{p_i} \subseteq Neig_{p_i}$ of its neighbors such that every neighbor $q \in S_{p_i}$ has a rank lower than $p_i$'s rank and $p_i \cup S_{p_i}$ is a maximal clique, i.e., $(\forall q \in (Neig_{p_i}-LNeig_{p_i}), [q \in S_{p_i} \wedge (\forall s \in S_{p_i}, q \neq s \wedge q \in Neig_s)] \vee [q \not \in S_{p_i} \wedge (\exists s \in S_{p_i},q \not \in Neig_s)])$;
\item Every node $q$ selected by $p_i$ has elected $p_i$ iff $p_i$ is the local leader of highest rank in $q$'s neighborhood, i.e., $(\forall q \in S_{p_i}, [\forall s \in (Neig_q \cup \{q\}), rank(p_i) \prec rank(s)] \Rightarrow lead_q=\id_{p_i})$;
\item Every node selected by $p_i$ which has elected $p_i$ belongs to the clique $C_i$ of $p_i$, i.e., $(\forall q \in S_{p_i}, lead_q=\id_{p_i} \Rightarrow q \in C_{p_i})$.
\end{enumerate}
\end{definition}

\begin{definition}[Path]
In a graph $G=(V,E)$, the sequence of nodes $\mathcal{P}_{G}(x,y)=<p_0=x, p_1, \dots, p_k=y>$ is called a \emph{path} between $x,y \in V$ if $\forall i, 1 \leq i \leq k, (p_i,p_{i-1}) \in E$. The nodes $p_0$ and $p_k$ are termed as the \emph{extremities} of $\mathcal{P}$. The length of $\mathcal{P}$ is noted $|\mathcal{P}|=k$.
\end{definition}

\begin{definition}[Legitimate configuration]
\label{def:config_legitime}
Let $\mathcal{C}$ be the set of all possible configuration of the system. A configuration $\gamma \in \mathcal{C}$ is legitimate for Algorithm $\mathcal{SS-CMCP}$ iff every clique constructed by a local leader in $\gamma$ satisfies Definition~\ref{def:correct_clique}.
\end{definition}

\subsubsection{Proof assuming a weakly fair daemon}

In the following we consider that for each node $p \in V$ the input $dist_p$ is correct, i.e., $dist_p$ is equal to the distance (in hops) between $p$ and $r$ in $G$. We begin the proof by showing in the above theorem that in an illegitimate configuration of the system there exists a node which can execute an action of Algorithm $\mathcal{SS-CMCP}$.

\begin{theorem}
\label{thm:cmcp_enable_action}
Let the set of configurations $\mathcal{B} \subseteq \mathcal{C}$ such that every configuration $\gamma \in \mathcal{B}$ satisfies Definition~\ref{def:config_legitime}. $\forall \gamma \in (\mathcal{C}-\mathcal{B}), \exists p \in V$ such that $p$ is enabled in $\gamma$. 
\end{theorem}

\begin{proof}
Assume, by the contradiction, that $\exists \gamma \in (\mathcal{C}-\mathcal{B})$ such that $\forall p \in V$ no action of Algorithm~\ref{algo} is enabled at $p$ in $\gamma$. According to Definition~\ref{def:config_legitime}, this implies that there exists a local leader $p_i \in V$ such that its clique $C_i$ does not satisfy Definition~\ref{def:correct_clique}.\\
If Claim 1 of Definition~\ref{def:correct_clique} is not satisfied in $\gamma$ then this implies that there is at least two local leaders in $C_i$. By definition of a local leader (see Definition~\ref{def:local_leader}), there is a node $q$ in $C_i$, $q \neq p_i$, which satisfies Predicate $\neg Selected(q)$. This implies that $q$ has not been selected by $p_i$, i.e., $q \not \in S_{p_i}$ and $q \in C_{p_i}$. According to the formal description of Algorithm $\mathcal{SS-CMCP}$, Macro $Clique_{p_i}$ returns the selected neighbors of $p_i$ which has elected $p_i$. So, since $q \not \in Leader_{p_i}$ and $q \in C_{p_i}$ then we have $C_{p_i} \neq Leader_{p_i}$ and $C3$-action is enabled at $p_i$, a contradiction. If Claim 2 of Definition~\ref{def:correct_clique} is not satisfied in $\gamma$ then this implies that either $p_i$ has selected a subset of its neighbors $S_{p_i}$ which does not form a maximal subgraph, i.e., we have $(\exists q \in (Neig_{p_i}-LNeig_{p_i}), q \not \in S_{p_i} \wedge (\forall s \in S_{p_i}, q \in Neig_s))$, or the selected subset $S_{p_i}$ does not define with $p_i$ a complete subgraph, i.e., $(\exists q \in (Neig_{p_i}-LNeig_{p_i}), q \in S_{p_i} \wedge (\exists s \in S_{p_i}, q \not \in Neig_s))$. According to the formal description of Algorithm $\mathcal{SS-CMCP}$, a local leader computes its selected neighbors using Procedure $Clique\_temp()$. So, we have $S_{p_i} \neq Clique\_temp()$ and $C1$-action is enabled at $p_i$, a contradiction. If Claim 3 of Definition~\ref{def:correct_clique} is not satisfied in $\gamma$ then this implies that there exists a neighbor $q$ selected by $p_i$ such that $q$ has not elected $p_i$ while $p_i$ has the highest rank in $q$'s neighborhood, i.e., $(\exists q \in S_{p_i}, (\forall s \in (Neig_q \cup q), rank(p_i) \prec rank(s)) \wedge lead_q \neq \id_{p_i})$. According to the formal description of Algorithm $\mathcal{SS-CMCP}$, Macro $Leader_p$ returns the neighbor of $p$ which has the highest rank. Moreover, we have $Leader_q = \id_{p_i}$ since by hypothesis $p_i$ has the highest rank in $q$'s neighborhood. So, we have $lead_q \neq Leader_q$ and $C2$-action is enabled at $q$, a contradiction. If Claim 4 of Definition~\ref{def:correct_clique} is not satisfied in $\gamma$ then this implies that there exists a selected neighbor $q$ of $p_i$ which does not belong to $C_i$ while $q$ has elected $p_i$, i.e., $(\exists q \in S_{p_i}, lead_q=\id_{p_i} \wedge q \not \in C_{p_i})$. According to the formal description of Algorithm $\mathcal{SS-CMCP}$, Macro $Clique_{p_i}$ returns the selected neighbors of $p_i$ which have elected $p_i$. So, we have $C_{p_i} \neq Clique_{p_i}$ and $C3$-action is enabled at $p_i$, a contradiction.
\end{proof}

To show the convergence of Algorithm $\mathcal{SS-CMCP}$ to a legitimate configuration, we now prove several sub-lemmas allowing to show that Algorithm $\mathcal{SS-CMCP}$ constructs a partition of correct cliques following the rank of the cliques (see Lemma~\ref{lem:construct_clique_prio}).

\begin{lemma}
\label{lem:N-action}
After executing $N$-action at any node $p \in V$, $N$-action is disabled at $p$.
\end{lemma}

\begin{proof}
Assume, by the contradiction, that $N$-action is enabled at any node $p \in V$ after its execution. If $p$ can execute $N$-action again then this implies that we have $N_p \neq Neig_p$ or $d_p \neq dist_p$ which is due to a modification in $p$'s neighborhood or a fault. This is a contradiction because we assume a static graph $G=(V,E)$ and a system execution without faults until reaching a legitimate configuration starting from an arbitrary configuration. 
\end{proof}

In the following, we note $\overline{S} \subseteq V$ the set of nodes in $\gamma \in \mathcal{C}$ such that every node $p \in \overline{S}$ is not selected by a neighbor of rank higher than $rank(p)$, i.e., $\overline{S}$ contains the set of local leaders in $\gamma$.

\begin{rem}
\label{rem:nodes_rank_clique}
A local leader $p_i \in \overline{S}$ can only select a node $p$ in its neighborhood such that $rank(p_i) \prec rank(p)$.
\end{rem}

\begin{proof}
According to the formal description of Algorithm $\mathcal{SS-CMCP}$, Macro $LNeig_{p_i}$ returns the neighbors $p$ of $p_i$ such that $rank(p_i) \prec rank(p)$. Moreover, Procedure $Clique\_temp()$ chooses nodes in the neighborhood of $p_i$ which are not included in the set given by Macro $LNeig_{p_i}$ (see line 2 of Procedure $Clique\_temp()$).
\end{proof}

\begin{lemma}
\label{lem:C1-action}
When $C1$-action is enabled at $p_i \in \overline{S}$, it remains enabled until $p_i$ executes it or $p_i \not \in \overline{S}$.
\end{lemma}

\begin{proof}
Let $\gamma \mapsto \gamma'$ be a step. Assume, by the contradiction, that $C1$-action is enabled at $p_i \in \overline{S}$ in $\gamma$ and not in $\gamma'$ (i.e., $S_{p_i}=Clique\_temp()$ in $\gamma'$) but $p_i$ did not execute $C1$-action in $\gamma \mapsto \gamma'$. According to the hypothesis of the lemma, we assume that $p_i \in \overline{S}$ in $\gamma'$, so we have $\neg Selected(p_i)$ in $\gamma'$. Since $p_i$ did not move in $\gamma \mapsto \gamma'$ and the variable $S_{p_i}$ can only be modified locally by $p_i$ by executing $C1$-action, we have $S_{p_i} \neq Clique\_temp()$ at $p_i$ in $\gamma'$, a contradiction.
\end{proof}

\begin{lemma}
\label{lem:select_clique}
The node $p_i \in \overline{S}$ of highest rank selects the maximal subset of its neighbors which can belong to its clique $C_i$ if $C_i$ does not satisfy Claim 2 of Definition~\ref{def:correct_clique}.
\end{lemma}

\begin{proof}
According to the formal description of Algorithm $\mathcal{SS-CMSP}$, a local leader executes $C1$-action to select the maximal subset of its neighbors which can belong to its clique. Assume, by the contradiction, that the node $p_i \in \overline{S}$ of highest rank does not select the maximal subset of its neighbors to belong to its clique $C_i$ while $C_i$ does not satisfy Claim 2 of Definition~\ref{def:correct_clique}. That is, $C1$-action is disabled or it is not the enabled action of highest priority at $p_i$.\\
We first show that $C1$-action is enabled at $p_i$. By definition of $\overline{S}$, we have $\neg Selected(p_i)$ at $p_i$. Moreover, Procedure $Clique\_temp()$ chooses in a deterministic greedy manner a maximal subset of $p_i$'s neighbors which define with $p_i$ a complete subgraph, i.e., satisfying $(\forall q \in Neig_{p_i}, [q \in S_{p_i} \wedge (\forall s \in S_{p_i}, q \neq s \wedge q \in Neig_s)] \vee [q \not \in S_{p_i} \wedge (\exists s \in S_{p_i}, q \not \in Neig_s)])$. Since Claim 2 of Definition~\ref{def:correct_clique} is not satisfied, we have two cases: (i) either $p_i$ has not selected a subset of neighbors defining with $p_i$ a complete subgraph, i.e., we have $(\exists q \in S_{p_i}, (\exists s \in S_{p_i}, q \neq s \wedge q \not \in Neig_s))$, or (ii) the subset of neighbors selected by $p_i$ is not maximal, i.e., we have $(\exists q \in (Neig_{p_i}-S_{p_i}), (\forall s \in S_{p_i}, q \in Neig_s))$. Thus, we have $S_{p_i} \neq Clique\_temp()$ and $C1$-action is enabled at $p_i$, a contradiction.\\
We must show that $C1$-action is the enabled action of highest priority at $p_i$. If $C1$-action is not the enabled action of highest priority at $p_i$ then this implies that $N$-action is always enabled. According to Lemma~\ref{lem:N-action}, after executing $N$-action it is not enabled at $p_i$, a contradiction. So, $N$-action is disabled at $p_i$. Moreover, according to Lemma~\ref{lem:C1-action} $C1$-action is enabled at $p_i \in \overline{S}$ until it is executed.
\end{proof}

\begin{lemma}
\label{lem:C2-action}
When $C2$-action is enabled at $p \in (V-\overline{S})$, it remains enabled until $p$ executes it or $p \in \overline{S}$.
\end{lemma}

\begin{proof}
Let $\gamma \mapsto \gamma'$ be a step. Assume, by the contradiction, that $C2$-action is enabled at $p \in (V-\overline{S})$ and not in $\gamma'$ (i.e., $lead_p=Leader_p$ in $\gamma'$) but $p$ did not execute $C2$-action in $\gamma \mapsto \gamma'$. According to the hypothesis of the lemma, we assume that $p \in (V-\overline{S})$ in $\gamma'$, so we have $Selected(p)$ in $\gamma'$. Since $p$ did not move in $\gamma \mapsto \gamma'$ and the variable $lead_p$ can only be modified locally by $p$ by executing $C2$-action (note that $C1$-action is disabled at $p$ because we have $Selected(p)$), we have $lead_p \neq Leader_p$ at $p$ in $\gamma'$. So, $C2$-action is enabled at $p$ in $\gamma'$, a contradiction.
\end{proof}

\begin{lemma}
\label{lem:elect_clique}
In any configuration $\gamma \in (\mathcal{C}-\mathcal{B})$, the nodes selected by the node $p_i \in \overline{S}$ of highest rank in $\gamma$ elect $p_i$ if the clique $C_i$ constructed by $p_i$ does not satisfy Claim 3 of Definition~\ref{def:correct_clique} in $\gamma$.
\end{lemma}

\begin{proof}
According to the formal description of Algorithm $\mathcal{SS-CMCP}$, a node executes $C2$-action to elect among its neighbors the local leader of highest rank which has selected it. Since the clique $C_i$ of $p_i$ does not satisfy Claim 3 of Definition~\ref{def:correct_clique}, there is a node $p$ selected by the local leader $p_i \in \overline{S}$ of highest rank which has not elected $p_i$ in $\gamma$. Assume, by the contradiction, that $p$ does not elect $p_i$. That is, $C2$-action is disabled or it is not the enabled action of highest priority at $p$ in $\gamma$.\\
We first show that $C2$-action is enabled at $p$ in $\gamma$. Since $p$ is selected by $p_i$ we have $Selected(p)$ satisfied at $p$. Assume, by the contradiction, that $C2$-action is disabled at $p$. According to the hypothesis of the lemma, we assume that we have $lead_p \neq \id_{p_i}$ at $p$. According to the formal description of Algorithm $\mathcal{SS-CMCP}$, Macro $Leader_p$ returns the identifier of the local leader in $p$'s neighborhood of highest rank which has selected $p$, i.e., by hypothesis of the lemma $Leader_p$ returns $\id_{p_i}$. Thus, we have $lead_p \neq Leader_p$ and $C2$-action is enabled at $p$ in $\gamma$, a contradiction.\\
We must show that $C1$-action is the enabled action of highest priority at $p$. If $C2$-action is not the enabled action of highest priority at $p$ then this implies that $N$-action or $C1$-action are always enabled. According to Lemma~\ref{lem:N-action}, after executing $N$-action it is not enabled at $p$, a contradiction. So, $N$-action is disabled at $p$. Moreover, Predicate $Selected(p)$ is satisfied at $p$ since it is selected by the local leader $p_i$ and $C1$-action is disabled at $p$, a contradiction. Moreover, according to Lemma~\ref{lem:C2-action} $C2$-action is enabled at $p$ until it is executed.
\end{proof}

\begin{rem}
\label{rem:elect_clique}
In any configuration $\gamma \in \mathcal{C}$, any node $p \in V$ can belong to at most a single clique.
\end{rem}

\begin{proof}
This comes from the fact that in a configuration $\gamma \in \mathcal{C}$ any node $p$ elects a single local leader using its local variable $lead_p$ either by executing $C1$-action if $p$ is a local leader or by executing $C2$-action otherwise.
\end{proof}

\begin{lemma}
\label{lem:C3-action}
When $C3$-action is enabled at $p_i \in \overline{S}$, it remains enabled until $p_i$ executes it unless $p_i \not \in \overline{S}$ or $C2$-action is enabled.
\end{lemma}

\begin{proof}
Let $\gamma \mapsto \gamma'$ be a step. Assume, by the contradiction, that $C3$-action is enabled at $p_i \in \overline{S}$ and not in $\gamma'$ (i.e., $C_{p_i}=Clique_{p_i}$ in $\gamma'$) but $p_i$ did not execute $C3$-action in $\gamma \mapsto \gamma'$. According to the hypothesis of the lemma, we assume that $p_i \in \overline{S}$ in $\gamma'$, so we have $Selected(p_i) \wedge S_{p_i}=Clique\_temp()$ in $\gamma'$. Since $p_i$ did not move in $\gamma \mapsto \gamma'$ and the variable $C_{p_i}$ can only be modified locally by $p_i$ by executing $C3$-action, we have $C_{p_i} \neq Clique_{p_i}$ at $p_i$ in $\gamma'$. So, $C3$-action is enabled at $p_i$ in $\gamma'$, a contradiction.
\end{proof}

\begin{lemma}
\label{lem:update_clique}
In any configuration $\gamma \in (\mathcal{C}-\mathcal{B})$, the node $p_i \in \overline{S}$ of highest rank updates the set of nodes included in its clique $C_i$ if $C_i$ satisfies Claims 2 and 3 of Definition~\ref{def:correct_clique} but not Claim 4 of Definition~\ref{def:correct_clique}.
\end{lemma}

\begin{proof}
According to the formal description of Algorithm $\mathcal{SS-CMCP}$, a local leader executes $C3$-action to updates the maximal subset of its neighbors which belongs to its clique $C_i$. Since the clique $C_i$ of $p_i$ satisfies Claims 2 and 3 of Definition~\ref{def:correct_clique} but not Claim 4 of Definition~\ref{def:correct_clique}, there is a neighbor $p$ selected by $p_i$ which has elected $p_i$ but $p_i$ does not consider that $p$ is part of $C_i$. Assume, by the contradiction, that the node $p_i \in \overline{S}$ of highest rank does not updates the maximal subset of its neighbors which belong to its clique $C_i$ while its clique $C_i$ does not satisfy Claim 4 of Definition~\ref{def:correct_clique}. That is, $C3$-action is disabled or it is not the enabled action of highest priority at $p_i$ in $\gamma \in (\mathcal{C}-\mathcal{B})$.\\
We first show that $C3$-action is enabled at $p_i$ in $\gamma$. By definition of $\overline{S}$, we have $\neg Selected(p_i)$. According to the hypothesis of the lemma, we have $S_{p_i}=Clique\_temp()$ since $p_i$ has selected the subset of its neighbors which can belong to its clique $C_i$. Since Claim 4 of Definition~\ref{def:correct_clique} is not satisfied, there is a neighbor $q$ of $p_i$ which has elected $p_i$ but $q$ does not belong to $C_i$, i.e., we have $lead_q=\id_{p_i} \wedge q \not \in C_{p_i}$. According to the formal description of Algorithm $\mathcal{SS-CMCP}$, Macro $Clique_{p_i}$ returns the set of neighbors selected by $p_i$ which have elected $p_i$. So, $q$ belongs to the set given by Macro $Clique_{p_i}$ since we have $lead_q=\id_{p_i}$ at $q$ in $\gamma$. Thus, we have $C_{p_i} \neq Clique_{p_i}$ and $C3$-action is enabled at $p_i$ in $\gamma$, a contradiction.\\
We must show that $C3$-action is the enabled action of highest priority at $p_i$. If $C3$-action is not the enabled action of highest priority at $p_i$ then this implies that $N$-action, $C1$-action or $C2$-action are always enabled. According to Lemma~\ref{lem:N-action}, after executing $N$-action it is not enabled at $p_i$, a contradiction. So, $N$-action is disabled at $p_i$. Predicate $Selected(p)$ is not satisfied at $p_i$ since $p_i \in \overline{S}$, so $C2$-action is disabled at $p_i$, a contradiction. Moreover, $S_{p_i}=Clique\_temp()$ by hypothesis so $C1$-action is disabled at $p_i$, a contradiction. Finally, according to Lemma~\ref{lem:C3-action} $C3$-action is enabled at $p_i \in \overline{S}$ until it is executed.
\end{proof}

\begin{lemma}
\label{lem:construct_prio}
Let $p_i$ and $p_j$ be two local leaders such that $rank(p_i) \prec rank(p_j)$. The construction by $p_j$ of the clique $C_j$ cannot prevent the construction by $p_i$ of the clique $C_i$.
\end{lemma}

\begin{proof}
First of all, according to the formal description of Algorithm $\mathcal{SS-CMCP}$ $N$-action is executed at any node independently from the construction of the cliques to enable the computation of the 2-neighborhood at each node. Moreover, $C1$-action and $C3$-action are executed independently at any local leader, so a local leader cannot prevent another local leader to execute these actions. Finally, we have to consider the execution of $C2$-action at a node selected by several local leaders. Let $q$ be a node selected by two local leaders $p_i$ and $p_j$ such that $rank(p_i) \prec rank(p_j)$. Assume, by the contradiction, that $p_j$ prevents $q$ to join the clique $C_i$ constructed by $p_i$. This implies that $q$ cannot execute $C2$-action to elect $p_i$, a contradiction according to Lemma~\ref{lem:elect_clique}.
\end{proof}

\begin{lemma}
\label{lem:construct_clique_prio}
Starting from an arbitrary configuration, the local leader $p_i$ of highest rank can construct its clique $C_i$ if $C_i$ does not satisfy Definition~\ref{def:correct_clique}.
\end{lemma}

\begin{proof}
From Lemmas~\ref{lem:select_clique},~\ref{lem:elect_clique} and~\ref{lem:update_clique}, we have that the clique $C_i$ of the local leader $p_i$ of highest rank is constructed such that Claims 2 to 4 of Definition~\ref{def:correct_clique} are satisfied.\\
Finally we consider Claim 1 of Definition~\ref{def:correct_clique}. Assume, by the contradiction, that the constructed clique $C_i$ contains more than a single local leader. By Definition~\ref{def:local_leader}, there is a node $q$ in $C_i$, $q \neq p_i$, (i.e., $q \in C_{p_i}$) which satisfies Predicate $\neg Selected(q)$. This implies that $q$ has not been selected by $p_i$, i.e., $q \not \in S_{p_i}$. Thus, by Lemma~\ref{lem:update_clique} $p_i$ executes $C3$-action since $C_i$ does not satisfy Claim 4 of Definition~\ref{def:correct_clique}, a contradiction.\\
Finally, according to Lemma~\ref{lem:construct_prio} the construction of the clique $C_i$ by $p_i$ cannot be prevented by any other local leader since $p_i$ is the local leader of highest rank.
\end{proof}

We show in the following that Algorithm $\mathcal{SS-CMCP}$ reaches a legitimate configuration (Definition~\ref{def:config_legitime}) in finite time starting from an arbitrary configuration.

\begin{lemma}
\label{lem:correct_clique_round}
Starting from an arbitrary configuration, the local leader of highest rank constructs its clique in at most $O(1)$ (asynchronous) rounds if its clique does not satisfy Definition~\ref{def:correct_clique}.
\end{lemma}

\begin{proof}
Let $p_i \in \overline{S}$ be the local leader of highest rank whose clique $C_i$ does not satisfy Definition~\ref{def:correct_clique}. According to Lemma~\ref{lem:construct_clique_prio}, $p_i$ constructs its clique $C_i$ in order to satisfy Definition~\ref{def:correct_clique}.

First of all, note that if we have $N_p \neq Neig_p$ at a node $p \in V$ then $N$-action is enabled at $p$ in round 0. Therefore, since the daemon is weakly fair and according to Lemma~\ref{lem:N-action} in the first configuration of round 1 we have $N_p=Neig_p$ at every node $p \in V$.

In the first configuration of round 1, $C1$-action is the enabled action of highest priority at $p_i$. Since the daemon is weakly fair and according to Lemma~\ref{lem:C1-action} in the first configuration of round 2 we have $S_{p_i}=Clique\_temp()$ and $lead_{p_i}=\id_{p_i}$ at $p_i$. In the second configuration of round 1, every neighbor $q$ of $p_i$ such that $q \in S_{p_i}$ satisfies $Selected(q)$. If $lead_q \neq \id_{p_i}$ then $C2$-action is the enabled action of highest priority at $q$. Since the daemon is weakly fair and according to Lemma~\ref{lem:C2-action} every such neighbor $q$ executes $C2$-action to elect $p_i$, which is the local leader of highest rank in the neighborhood of $q$. Thus, in the first configuration of round 2 we have $S_q=C_q=\emptyset$, and $lead_q=\id_{p_i}$ at $q$. In the first configuration of round 2, $C3$-action is the enabled action of highest priority at $p_i$. Since the daemon is weakly fair and according to Lemma~\ref{lem:C3-action} in the first configuration of round 3 we have $C_{p_i}=Clique_{p_i}$ at $p_i$. Therefore, $p_i$ has constructed its clique $C_i$ in $O(1)$ rounds.
\end{proof}

\begin{lemma}
\label{lem:config_legitime_round}
Starting from any configuration in which for each node $p \in V$ the input $dist_p$ is correct, Algorithm $\mathcal{SS-CMCP}$ reaches a configuration satisfying Definition~\ref{def:config_legitime} in at most $O(\min(n_c \times Diam, n))$ (asynchronous) rounds, with $n_c$ the maximum number of cliques at any distance from $r$ in $G$, $Diam$ the diameter of $G$, and $n$ the number of nodes in $G$. Moreover, $O(\Delta \log(n))$ bits of memory are necessary at each node, with $\Delta$ the maximum degree of a node in $G$.
\end{lemma}

\begin{proof}
In the following, we define by $p_i^k$ a local leader $p_i \in \overline{S}$ at distance $k$ (in hops) from the root node $r$.

We first show by induction on the distances in $G$ the following proposition: in at most $O(n_k)$ rounds every local leader $p_i^k, 1 \leq i \leq n_k$ at distance $k$ from $r$ has constructed its clique $C_i$ satisfying Definition~\ref{def:correct_clique}, with $n_k$ the number of maximal cliques constructed at distance $k$.

In base case $k=0$. We must verify the proposition only at $r$ since there is no other local leader at distance 0 from $r$. According to Lemma~\ref{lem:correct_clique_round} in $O(1)$ rounds $r$ has constructed its clique, which verifies the proposition since $n_0=1$.\\
Induction case: We assume the proposition is verified for every local leader at distance $k-1$ from $r$ in $G$. We have to show the proposition is also verified for every local leader at distance $k$ from $r$. Consider the local leaders $p_i^k$ at distance $k$ from $r$, with $1 \leq i \leq n_k$, following the order of their rank from the highest to the lowest. We can apply iteratively Lemmas~\ref{lem:construct_clique_prio} and~\ref{lem:correct_clique_round} to show that each $p_i^k$ constructs its clique in $O(1)$ rounds. Therefore, in at most $O(n_k)$ rounds the proposition is verified at every local leader at distance $k$ from $r$.\\
Since there are at most $Diam+1$ layers with local leaders, in at most $O(\sum_{k=0}^{Diam} n_k) \leq O(n_c \times Diam)$ rounds the proposition is verified at every local leader, with $n_c=\max_{0 \leq k \leq Diam} n_k$. Moreover, we can observe that we cannot have more than $n$ cliques in any clique partition. Therefore, in at most $O(\min(n_c \times Diam, n))$ rounds the proposition is verified at every local leader.

We can observe that in the proposition used for the above induction proof every clique constructed by a local leader satisfies Definition~\ref{def:correct_clique}. Therefore, the configuration $\gamma$ reached by Algorithm $\mathcal{SS-CMCP}$ in $O(\min(n_c \times Diam, n))$ rounds satisfies Definition~\ref{def:config_legitime}.

Finally, according to the formal description of Algorithm $\mathcal{SS-CMCP}$ at any node $p \in V$ the variables $lead_p$ and $d_p$ are of size $O(\log(n))$ bits since they store a node identifier and a distance respectively of at most $n$ states. Moreover, the variables $N_p$, $S_p$ and $C_p$ store a subset of neighbors identifier composed of at most $\Delta$ elements leading to variables of size $O(\Delta \log(n))$ bits.
\end{proof}

Finally, we show below that any legitimate configuration reached by Algorithm $\mathcal{SS-CMCP}$ is a terminal configuration which defines a solution to the Connected Minimal Clique Partition problem.

\begin{lemma}
\label{lem:no_enable_action}
In any configuration $\gamma \in \mathcal{C}$, for every node $p$ which belongs to a clique $C_i$ satisfying Definition~\ref{def:correct_clique} in $\gamma$ no action of Algorithm~\ref{algo} is enabled at $p$.
\end{lemma}

\begin{proof}
Assume, by the contradiction, that there exists a configuration $\gamma \in \mathcal{C}$ such that there exists a node $p$ in a clique $C_i$ satisfying Definition~\ref{def:correct_clique} with an enabled action of Algorithm~\ref{algo} at $p$.

Let $p_i$ be the local leader of the clique $C_i$ in the following. If $N$-action is enabled at $p$ then $N_p \neq Neig_p$ or $d_p \neq dist_p$ and $p$ can execute $N$-action in step $\gamma \mapsto \gamma'$. In configuration $\gamma'$, we must consider two cases: either Definition~\ref{def:correct_clique} is not satisfied in $\gamma'$ a contradiction because this implies that $C_i$ did not satisfy Definition~\ref{def:correct_clique} in $\gamma$, otherwise Definition~\ref{def:correct_clique} is satisfied in $\gamma'$ and according to Lemma~\ref{lem:N-action} $N$-action is disabled, a contradiction.
If $C1$-action is enabled at $p$ then $p=p_i$ and we have $S_p \neq Clique\_temp()$. This implies that the nodes selected by $p$ does not form a maximal clique. That is, there exists a neighbor $q$ of $p$ such that $q \not \in S_p$ and $(\forall s \in S_p, q \in Neig_s)$, or $q \in S_p$ and $(\exists s \in S_p, q \not \in Neig_s)$. This is in contradiction with Claim 2 of Definition~\ref{def:correct_clique}.
If $C2$-action is enabled at $p$ then $p$ is not a local leader and we have $lead_p \neq Leader_p$. This implies that $p$ has elected $p_i$ but there exists a local leader $p_j$ in $p$'s neighborhood such that $rank(p_j) \prec rank(p_i)$, a contradiction with Claim 3 of Definition~\ref{def:correct_clique}.
If $C3$-action is enabled at $p$ then $p=p_i$ and we have $C_p \neq Clique_p$. This implies that there exists a node $q \in C_i, q \neq p_i,$ which has elected $p_i$ while $q \not \in C_p$, i.e., we have $lead_q=\id_{p_i} \wedge q \not \in C_p$. This is in contradiction with Claim 4 of Definition~\ref{def:correct_clique}.
\end{proof}

\begin{corollary}
\label{cor:no_enable_action}
In every configuration $\gamma \in \mathcal{B}$ satisfying Definition~\ref{def:config_legitime}, for every node $p \in V$ no action of Algorithm $\mathcal{CMCP}$ is enabled in $\gamma$.
\end{corollary}

\begin{proof}
According to Definition~\ref{def:config_legitime}, every clique constructed by a local leader in $\gamma \in \mathcal{B}$ satisfies Definition~\ref{def:correct_clique}. Therefore, we can apply Lemma~\ref{lem:no_enable_action} which shows the corollary.
\end{proof}

\begin{lemma}
\label{lem:sol_cmcp_config_legitime}
Let the set of configurations $\mathcal{B} \subseteq \mathcal{C}$ such that every configuration $\gamma \in \mathcal{B}$ satisfies Definition~\ref{def:config_legitime}. $\forall \gamma \in \mathcal{B}$, a connected minimal clique partition (Definition~\ref{def:cmcp}) is constructed in $\gamma$.
\end{lemma}

\begin{proof}
According to Definition~\ref{def:cmcp}, to prove the lemma we must show that the clique partition constructed in every configuration $\gamma \in \mathcal{B}$ is: (i) minimal for inclusion, and (ii) connected.

Consider first the minimality property of the clique partition. Assume, by the contradiction, that the first property is not satisfied in $\gamma \in \mathcal{B}$. This implies that if we take the cliques following their ranks from the highest to the lowest rank then there are two cliques $C_i$ and $C_j$ such that $C_i \cup C_j$ is a clique in $\gamma$. However, according to Remark~\ref{rem:total_order_clique} we have a total order on the cliques and the clique of highest rank, say $C_i$, is not a maximal clique. However, $C_i$ satisfies Definition~\ref{def:correct_clique} because $\gamma \in \mathcal{B}$. So, according to Claim 2 of Definition~\ref{def:correct_clique} $C_i$ is a maximal clique, a contradiction.

Consider now the connectivity property of the clique partition. Assume, by the contradiction, that the clique partition constructed in $\gamma \in \mathcal{B}$ is not connected. This implies that the graph $G_c=(V_c, E_c)$ induced by the non trivial cliques is not connected in $\gamma$. Thus, there exists a local leader $p \in V_c$ such that there is no path $\mathcal{P}_{G_c}(p, r)$ between $p$ and $r$ in $G_c$. Consider the local leader $p_i$ of highest rank in $\gamma$ such that $\not \exists \mathcal{P}_{G_c}(p_i, r)$ in $G_c$. According to Remark~\ref{rem:nodes_rank_clique}, a correct clique can only contain nodes with a rank lower than the rank of the local leader of the clique. So, by definition of ranks we have only to consider the shortest paths between $p_i$ and $r$ in $G$. Every shortest path $\mathcal{P}_G(p_i,r)$ in $G$ can be decomposed in three parts: $\mathcal{P}_G^1(p_i,r)$ containing the nodes in $G_c$, $\mathcal{P}_G^2(p_i,r)$ containing the nodes in $(G - G_c)$, and $p_i \in G_c$. In every shortest path $\mathcal{P}_G(p_i,r)$, any node $p_j \in \mathcal{P}_G^2(p_i,r)$ is a local leader of its trivial clique $C_j$ because $p_j \not \in G_c$. Since $\gamma \in \mathcal{B}$, $C_j$ is a maximal clique according to Claim 2 of Definition~\ref{def:correct_clique}. However, there is a neighbor $q$ of $p_j$ in $\gamma$ such that either $q \in \mathcal{P}_G^2(p_i,r)$ or $q=p_i$. Thus, we have $rank(p_j) \prec rank(q)$, a contradiction with Claim 2 of Definition~\ref{def:correct_clique} since $C_j$ is not maximal.
\end{proof}

\begin{theorem}
Algorithm $\mathcal{SS-CMCP}$ is a self-stabilizing algorithm for Specification~\ref{spec:cmcp} under a weakly fair distributed daemon.
\end{theorem}

\begin{proof}
We have to show that starting from any configuration the execution of Algorithm $\mathcal{SS-CMCP}$ verifies the two conditions of Specification~\ref{spec:cmcp}.

According to Theorem~\ref{thm:cmcp_enable_action}, Lemma~\ref{lem:config_legitime_round} and Corollary~\ref{cor:no_enable_action}, from any configuration Algorithm $\mathcal{SS-CMCP}$ reaches a configuration $\gamma \in \mathcal{C}$ in finite time and $\gamma$ is a terminal configuration, which verifies Condition 1 of Specification~\ref{spec:cmcp}. Moreover, according to Lemma~\ref{lem:sol_cmcp_config_legitime} the terminal configuration $\gamma$ reached by Algorithm $\mathcal{SS-CMCP}$ satisfies Definition~\ref{def:cmcp}, which verifies Condition 2 of Specification~\ref{spec:cmcp}.
\end{proof}

Finally, from an arbitrary configuration we can establish the following corollary according to Lemma~\ref{lem:config_legitime_round}.

\begin{corollary}
\label{cor:config_legitime_round_avec_BFS}
Starting from an arbitrary configuration, the fair composition of Algorithm $\mathcal{SS-CMCP}$ and Algorithm $\mathcal{A_{BFS}}$ reach a configuration satisfying Definition~\ref{def:config_legitime} in at most $O(T_{\mathcal{BFS}}+\min(n_c \times Diam, n))$ (asynchronous) rounds, with $T_{\mathcal{BFS}}$ the round complexity of self-stabilizing algorithm $\mathcal{A_{BFS}}$ constructing a BFS tree, $n_c$ the maximum number of cliques at any distance from $r$ in $G$, $Diam$ the diameter of $G$, and $n$ the number of nodes in $G$.
\end{corollary}

\subsubsection{Proof assuming an unfair daemon}

In the following, we prove that Algorithm $\mathcal{SS-CMCP}$ is self-stabilizing under an unfair daemaon by bounding the number of steps needed to reach a legitimate configuration.

\begin{lemma}
\label{lem:N-action_steps}
In an execution, every node $p \in V$ can execute $N$-action at most once.
\end{lemma}

\begin{proof}
According to Lemma~\ref{lem:N-action} if $N$-action is enabled at a node $p \in V$ in the initial configuration then it becomes disabled after its execution at $p$.
\end{proof}

In the following we consider that for each node $p \in V$ the input $dist_p$ is correct, i.e., $dist_p$ is equal to the distance (in hops) between $p$ and $r$ in $G$.

\begin{definition}[Priority level]
The \emph{priority level} of any node $p \in V$ is equal to the number of nodes $q \in V$ such that $rank(q) \prec rank(p)$ in $G$. The priority level of a clique is defined by the priority level of its local leader.
\end{definition}

\begin{lemma}
Let $C_i$ be a correct clique (Definition~\ref{def:correct_clique}) of priority level $i, 0 \leq i \leq n-1,$ which belongs to a legitimate configuration $\gamma \in \mathcal{B}$. In an execution, a correct clique $C_i$ may not satisfy Definition~\ref{def:correct_clique} at most $i$ times.
\end{lemma}

\begin{proof}
According to Remark~\ref{rem:total_order_clique}, in any clique partition the rank of the cliques define a total order. Moreover, according to Lemma~\ref{lem:construct_prio} the construction of any clique $C_j$ cannot prevent the construction of another clique $C_i$ if $rank(C_i) \prec rank(C_j)$. Thus, the construction of the clique $C_i$ of priority level $i$ can be prevented by at most $i$ cliques. However, as long as these $i$ cliques of rank higher than $C_i$ do not satisfy Definition~\ref{def:correct_clique} the construction of $C_i$ can be affected. Consider the following worst case scheduling. The cliques $C_j, 0 \leq j \leq i,$ are constructed following their rank from the lowest to the highest rank, and before the construction of a new clique $C_j, j<i,$ the construction of $C_i$ is performed again in order to satisfy Definition~\ref{def:correct_clique}. Thus, the construction of each clique $C_j, j<i,$ involves that $C_i$ does not satisfy Definition~\ref{def:correct_clique} and this situation happens at most $i$ times.
\end{proof}

According to the formal description of Algorithm $\mathcal{SS-CMCP}$, the construction of a correct clique $C_i$ is performed by executing $C1$-action and $C3$-action or $C2$-action at a node $p \in C_i$.

\begin{corollary}
\label{cor:other_actions_steps}
Let $C_i$ be a correct clique (Definition~\ref{def:correct_clique}) of priority level $i, 0 \leq i \leq n-1,$ which belongs to a legitimate configuration $\gamma \in \mathcal{B}$. In an execution, a node $p \in C_i$ can execute $C1$-action, $C2$-action and $C3$-action at most $i$ times.
\end{corollary}

\begin{lemma}
\label{lem:config_legitime_step}
From any configuration in which for each node $p \in V$ the input $dist_p$ is correct, at most $O(n^2)$ steps are needed by Algorithm $\mathcal{SS-CMCP}$ to reach a configuration satisfying Definition~\ref{def:config_legitime}, with $n$ the number of nodes in $G$.
\end{lemma}

\begin{proof}
First of all, according to Lemma~\ref{lem:N-action_steps} in an execution of Algorithm $\mathcal{SS-CMCP}$ $N$-action is executed at most $n$ times. Moreover, according to Corollary~\ref{cor:other_actions_steps} in an execution of Algorithm $\mathcal{SS-CMCP}$ a node $p \in V$ of priority level $i, 0 \leq i \leq n-1,$ can execute $C1$-action, $C2$-action and $C3$-action at most $i$ times. Moreover, a clique partition contains at most $n$ cliques. So, by summing up we have that in an execution of Algorithm $\mathcal{SS-CMCP}$ starting from any configuration in which the input $dist_p$ is correct for each node $p \in V$ $C1$-action, $C2$-action and $C3$-action are executed at most $\sum_{i=0}^{n-1}i=\frac{n(n+1)}{2}$ times.

Therefore, from any configuration in which the input $dist_p$ is correct for each node $p \in V$ Algorithm $\mathcal{SS-CMCP}$ executes at most $n+\frac{n(n+1)}{2}<O(n^2)$ steps to reach a configuration satisfying Definition~\ref{def:config_legitime}.
\end{proof}

Finally, from any configuration we can establish the following corollary according to Lemma~\ref{lem:config_legitime_step}.

\begin{corollary}
\label{cor:config_legitime_step_avec_BFS}
From any configuration, at most $O(ST_{\mathcal{BFS}} \times n^2)$ steps are needed by Algorithms $\mathcal{SS-CMCP}$ and $\mathcal{A_{BFS}}$ executed following a fair composition to reach a configuration satisfying Definition~\ref{def:config_legitime}, with $n$ the number of nodes in $G$ and $ST_{\mathcal{BFS}}$ the step complexity of self-stabilizing algorithm $\mathcal{A_{BFS}}$ constructing a BFS tree.
\end{corollary}

\section{Self-stabilizing Connected Vertex Cover}
\label{sec:cvc}

We define below an extension of the classical Vertex Cover problem, called Connected Vertex Cover problem.

\begin{definition}[2-approximation Connected Vertex Cover]
\label{def:cvc}
Let $G=(V,E)$ be any undirected graph. A vertex cover $S$ of the graph $G$ is \emph{connected} iff for any pair of node $u,v \in S$ there is a path between $u$ and $v$ in the graph induced by $S$. Moreover, $S$ is a 2-approximation Connected Vertex Cover, i.e., we have $|S| \leq 2|CVC^*|$ with $CVC^*$ an optimal solution for the Connected Vertex Cover.
\end{definition}

In~\cite{DelbotLP13}, Delbot \emph{et al.} presented a centralized optimization algorithm to solve the Connected Vertex Cover problem which uses a solution obtained for the Connected Minimal Clique Partition problem (see Definition~\ref{def:cmcp}). Given a solution $S$ for the Connected Minimal Clique Partition, the authors have shown in~\cite{DelbotLP13} that we can construct a solution $S'$ for the Connected Vertex Cover with an approximation ratio of 2 by selecting in $S'$ all the cliques in $S$ which are not \emph{trivial}, i.e., by selecting all the cliques composed of at least two nodes.  

In the following, we define in Specification~\ref{spec:cvc} the Self-stabilizing Connected Vertex Cover problem.

\begin{spec}[Self-stabilizing Connected Vertex Cover]
\label{spec:cvc}
Let $\mathcal{C}$ the set of all possible configurations of the system. An algorithm $\mathcal{A_{CVC}}$ solving the problem of constructing a stabilizing connected vertex cover satisfies the following conditions:
\begin{enumerate}
\item Algorithm $\mathcal{A}$ reaches a set of terminal configurations $\mathcal{T} \subseteq \mathcal{C}$ in finite time, and
\item Every configuration $\gamma \in \mathcal{T}$ satisfies Definition~\ref{def:cvc}.
\end{enumerate}
\end{spec}

\subsection{Related works}

The Vertex Cover problem is a classical optimization problem and many works have been devoted to this problem or to its variations. This problem is known to be APX-complete \cite{PapadimitriouY88} and not approximable within a factor of $10\sqrt{5}-21\approx 1.36067$ \cite{DinurS05}. Some very simple approximation algorithms gives a tight approximation ratio of $2$ \cite{GJ79,Vaz01,Savage82}. Despite a lot of works, no algorithm whose approximation ratio is bounded by a constant less than $2$ has been found and it is conjectured that there is no smaller \emph{constant} ratio unless $P=NP$ \cite{KhotR08}. Monien and Speckenmeyer \cite{Monien1985} and Bar-Yehuda and Even \cite{Bar-YehudaE85} proposed algorithms with an approximation ratio of $2-\frac{\ln \ln n}{\ln n}$, with $n$ the number of vertices of the graph and Karakostas \cite{Karakostas05} reduced this ratio to $2-\Theta (\frac{1}{\sqrt{\log{n}}})$.

From a self-stabilizing point of view, Kiniwa~\cite{Kiniwa05} proposed the first self-stabilizing algorithm for this problem which constructs a 2-approximate vertex cover in general networks with unique nodes identifier and under a fair distributed daemon. This algorithm is based on the construction of a maximal matching which allows to obtain a 2-approximation vertex cover by selecting the extremities of the matching edges. Turau \emph{et al.}~\cite{TurauH11} considered the same problem in anonymous networks and gave an 3-approximation algorithm under a distributed daemon. Since it is impossible to construct a maximal matching in an anonymous network, this algorithm establishes first a bicolored graph of the network allowing then to construct a maximal matching to obtain a vertex cover. Turau~\cite{Turau10} designed a self-stabilizing vertex cover algorithm with approximation ratio of 2 in anonymous networks under an unfair distributed daemon. This algorithm uses the algorithm in~\cite{TurauH11} executed several times on parts of the graph to improve the quality of the constructed solution.

For the Connected Vertex Cover problem, Savage in~\cite{Savage82} proposed a 2-approximation algorithm in general graphs based on the construction of a Depth First Search tree $T$ and selecting in the solution the nodes with at least a child in $T$. In 2010 Escoffier \emph{et al.}~\cite{EscoffierGouvesMonnot10} proved that the problem is NP-complete, even in bipartite graphs (whereas it is polynomial to construct a vertex cover in bipartite graphs), is polynomial in chordal graphs and can be approximated with better ratio than $2$ in several restricted classes of graphs.

To our knowledge, there exists no self-stabilizing algorithm for the Connected vertex cover problem. However, the approach proposed by Savage~\cite{Savage82} can be used to design a self-stabilizing algorithm. Indeed, any self-stabilizing algorithm performing a depth first search traversal of the graph (e.g., see~\cite{CollinD94,CournierDPV06,PetitV07}) executed in parallel with the algorithm described later in this section can be used to select the appropriate set of nodes in the solution. However, this does not enable to obtain the best complexity in terms of time. Although a low memory complexity of $O(\log(\Delta))$ bits per node is reached, this approach has a time complexity of $\Theta(n)$ rounds. Indeed, a low level of parallelism is reached because of the DFS traversal. In contrast, the self-stabilizing algorithm that we propose in this section is based on the algorithm presented in the previous section. Our solution has a better time complexity of $O(\min(n_c \times Diam,n))$ rounds because of the parallel construction of cliques. However, the memory complexity is $O(\Delta \log(n))$ bits per node.

\subsection{Self-stabilizing construction}

In this subsection, we present our self-stabilizing Connected Vertex Cover algorithm called $\mathcal{SS-CVC}$ which follows the approach given in~\cite{DelbotLP13}. A solution to the Connected Vertex Cover problem contains all the non trivial cliques of a Connected Minimal Clique Partition. We give in this section a self-stabilizing algorithm allowing to select the nodes of non trivial cliques, a formal description is given in Algorithm~\ref{algo2}. So, Algorithm $\mathcal{SS-CVC}$ is defined as a fair composition~\cite{Dolev2000} of Algorithms~\ref{algo} and~\ref{algo2} which are executed at each node $p \in V$.\\
Algorithm~\ref{algo2} takes in input at each node $p$ the local leader of $p$ and the set of nodes belonging to the maximal clique of $p$ given by Algorithm~\ref{algo} (i.e., variables $lead_p$ and $C_p$ of Algorithm~\ref{algo}) in case $p$ is a local leader. Moreover, in Algorithm~\ref{algo2} each node maintains a single boolean variable $In_p$. Any node $p$ belongs to the Connected Vertex Cover if and only if (1) either it is a local leader and its maximal clique is not trivial (i.e., $lead_p=\id_p$ and $|C_p|>1$), or (2) it is contained in a maximal clique constructed by a neighbor which is the local leader of $p$ (i.e., $lead_p \neq \id_p$). Predicate $InVC(p)$ is satisfied at each node $p$ if $p$ is part of the Connected Vertex Cover. Therefore, Algorithm~\ref{algo2} is composed of a single rule executed by each node $p \in V$ to correct the value of variable $In_p$ in order to be equal to the value of Predicate $InVC(p)$. So, a solution to the Connected Vertex Cover problem contains every node $p$ such that $In_p=true$.

\begin{algorithm}
\caption{\quad Self-Stabilizing Connected Vertex Cover algorithm for any $p \in V$\label{algo2}}
\smallskip
\begin{scriptsize}
{\bf Inputs:}\\
\hspace*{1cm}$\id_p$: unique identifier of $p$;\\
\hspace*{1cm}$lead_p$: leader of $p$ computed by Algorithm~\ref{algo};\\
\hspace*{1cm}$C_p$: maximal clique of $p$ computed by Algorithm~\ref{algo};\\
{\bf Variable:}\\
\hspace*{1cm}$In_p \in \{true, false\}$;

.\dotfill\ 

{\bf Predicate:}\\
\begin{tabular}{lll}
$InVC(p)$ & $\equiv$ & $(lead_p \neq \id_p \vee |C_p|>1)$\\
\end{tabular}\\
.\dotfill\

{\bf Action:}\\
\begin{tabular}{lllllll}
&& $VC$-$action$  & $::$ & $In_p \neq InVC(p)$ & $\to$ & $In_p:=InVC(p)$;\\
\end{tabular}
\end{scriptsize}
\end{algorithm}

\subsection{Correctness proof}

\begin{definition}[Legitimate configuration]
\label{def:config_legitime_cvc}
A configuration $\gamma \in \mathcal{C}$ is legitimate for Algorithm~\ref{algo2} iff for every node $p \in V$ we have $In_p=InVC(p)$.
\end{definition}

In the following we consider that Algorithm $\mathcal{SS-CMCP}$ is stabilized and we have correct inputs for $lead_p$ and $C_p$ at every node $p \in V$.

\begin{theorem}
\label{thm:cvc_enable_action}
Let the set of configurations $\mathcal{B} \subseteq \mathcal{C}$ such that every configuration $\gamma \in \mathcal{B}$ satisfies Definition~\ref{def:config_legitime_cvc}. $\forall \gamma \in (\mathcal{C}-\mathcal{B}), \exists p \in V$ such that $p$ is enabled in $\gamma$.
\end{theorem}

\begin{proof}
Assume, by the contradiction, that $\exists \gamma \in (\mathcal{C}-\mathcal{B})$ such that $\forall p \in V$ no action of Algorithm~\ref{algo2} is enabled at $p$ in $\gamma$. According to Definition~\ref{def:config_legitime_cvc}, this implies that there exists a node $p \in V$ such that $In_p \neq InVC(p)$. So, $VC$-action is enabled at $p$, a contradiction.
\end{proof}

\begin{lemma}
\label{lem:VC-action}
When $VC$-action is enabled at any $p \in V$, it remains enabled until $p$ executes it.
\end{lemma}

\begin{proof}
Let $\gamma \mapsto \gamma'$ be a step. Assume, by the contradiction, that $VC$-action is enabled at $p$ in $\gamma$ and not in $\gamma'$ (i.e., $In_p=InVC(p)$ in $\gamma'$) but $p$ did not execute $VC$-action in $\gamma \mapsto \gamma'$. Since $p$ did not move in $\gamma \mapsto \gamma'$ and the variable $In_p$ can only be modified locally by $p$ by executing $VC$-action, we have $In_p \neq InVC(p)$ at $p$ in $\gamma'$, a contradiction.
\end{proof}

\begin{lemma}
\label{lem:config_legitime_cvc_round}
Starting from any configuration satisfying Definition~\ref{def:config_legitime}, Algorithm~\ref{algo2} reaches a configuration satisfying Definition~\ref{def:config_legitime_cvc} in at most $O(1)$ (asynchronous) rounds. Moreover, $O(1)$ bits of memory are necessary at each node.
\end{lemma}

\begin{proof}
In any configuration satisfying Definition~\ref{def:config_legitime}, if we have $In_p \neq InVC(p)$ at a node $p \in V$ then $VC$-action is enabled at $p$ in round 0. Therefore, according to Lemma~\ref{lem:VC-action} in the first configuration $\gamma$ of round 1 we have $In_p=InVC(p)$ at every node $p \in V$. Moreover, this implies that $\gamma$ satisfies Definition~\ref{def:config_legitime_cvc}.

We can observe that Algorithm~\ref{algo2} maintains a single boolean variable $In_p$ at each node $p \in V$. So, $O(1)$ bits of memory are necessary at each node $p \in V$.
\end{proof}

From Corollary~\ref{cor:config_legitime_round_avec_BFS} and Lemma~\ref{lem:config_legitime_cvc_round}, we can establish the round complexity given in the following corollary.

\begin{corollary}
Starting from any configuration, the fair composition of Algorithms $\mathcal{A_{BFS}}$ and $\mathcal{SS-CVC}$ reach a configuration satisfying Definition~\ref{def:config_legitime_cvc} in at most $O(T_{\mathcal{BFS}}+\min(n_c \times Diam, n)+1)$ (asynchronous) rounds, with $T_{\mathcal{BFS}}$ the round complexity of self-stabilizing algorithm $\mathcal{A_{BFS}}$ constructing a BFS tree, $n_c$ the maximum number of cliques at any distance from $r$ in $G$, $Diam$ the diameter of $G$, and $n$ the number of nodes in $G$. Moreover, $O(\Delta \log(n))$ bits of memory are necessary at each node, with $\Delta$ the maximum degree of a node.
\end{corollary}

\begin{lemma}
\label{lem:config_legitime_cvc_step}
Starting from any configuration satisfying Definition~\ref{def:config_legitime}, at most $O(n)$ steps are needed by Algorithm~\ref{algo2} to reach a configuration satisfying Definition~\ref{def:config_legitime_cvc}, with $n$ the number of nodes in $G$.
\end{lemma}

\begin{proof}
In any configuration $\gamma$ satisfying Definition~\ref{def:config_legitime}, if we have $In_p \neq InVC(p)$ at a node $p \in V$ in $\gamma$ then $VC$-action is enabled at $p$ in $\gamma$. According to Lemma~\ref{lem:VC-action}, $VC$-action is enabled at $p$ until it is executed. $VC$-action can be enabled at every node $p \in V$ in $\gamma$. So, each node $p$ can execute $VC$-action because it is the action of highest priority at $p$ since Algorithm~\ref{algo2} is composed of a single action. Thus, after at most $O(n)$ steps Algorithm~\ref{algo2} has reached a configuration $\gamma'$ such that we have $In_p=InVC(p)$ at every node $p \in V$ in $\gamma'$. Moreover, this implies that Definition~\ref{def:config_legitime_cvc} is satisfied in $\gamma'$.
\end{proof}

From Corollary~\ref{cor:config_legitime_step_avec_BFS} and Lemma~\ref{lem:config_legitime_cvc_step}, we can establish the step complexity of Algorithm $\mathcal{SS-CVC}$ given in the following corollary.

\begin{corollary}
Starting from any configuration, in at most $O(ST_{\mathcal{BFS}} \times n^3)$ steps are needed by Algorithms $\mathcal{A_{BFS}}$ and $\mathcal{SS-CVC}$ executed following a fair composition to reach a configuration satisfying Definition~\ref{def:config_legitime_cvc}, with $ST_{\mathcal{BFS}}$ the step complexity of self-stabilizing algorithm $\mathcal{A_{BFS}}$ constructing a BFS tree and $n$ the number of nodes in $G$.
\end{corollary}

\begin{lemma}
\label{lem:no_enable_action_cvc}
In every configuration $\gamma \in \mathcal{B}$ satisfying Definition~\ref{def:config_legitime_cvc}, for every node $p \in V$ no action of Algorithm~\ref{algo2} is enabled in $\gamma$.
\end{lemma}

\begin{proof}
Assume, by the contradiction, that there exists a configuration $\gamma \in \mathcal{B}$ such that there exists a node $p \in V$ with an enabled action of Algorithm~\ref{algo2}. According to the formal description of Algorithm~\ref{algo2}, the algorithm is only composed of $VC$-action. This implies that we have $In_p \neq InVC(p)$ at $p$ in $\gamma$. However, we have $In_p=InVC(p)$ at every node $p \in V$ because $\gamma \in \mathcal{B}$, a contradiction.
\end{proof}

\begin{lemma}
\label{lem:sol_cvc_config_legitime}
Let the set of configurations $\mathcal{B} \subseteq \mathcal{C}$ such that every configuration $\gamma \in \mathcal{B}$ satisfies Definition~\ref{def:config_legitime_cvc}. $\forall \gamma \in \mathcal{B}$, a $2$-approximated Connected Vertex Cover (Definition~\ref{def:cvc}) is constructed in $\gamma$.
\end{lemma}

\begin{proof}
According to Definition~\ref{def:cvc}, to prove the lemma we must show that the solution $S$ constructed in every configuration $\gamma \in \mathcal{B}$ is: (i) a vertex cover of $G$, (ii) connected, and (iii) a 2-approximation from an optimal solution.

According to Specification~\ref{spec:cmcp}, Algorithm~\ref{algo2} takes in input a Connected Minimal Clique partition. Consider the first property. In any configuration $\gamma \in \mathcal{B}$, according to Algorithm~\ref{algo2} only the nodes which belong to a non trivial clique are included in the constructed solution $S$. Assume, by the contradiction, that $S$ is not a vertex cover of $G$. This implies that there exists an edge between two trivial cliques $C_i$ and $C_j$ of the clique partition given in input. So, the clique partition given in input is not minimal since we can construct the maximal clique $C_i \cup C_j$, a contradiction with Specification~\ref{spec:cmcp}. So, the set of trivial cliques forms an independent set and all the edges of the graph are covered by the nodes in $S$. Consider the second property. According to Specification~\ref{spec:cmcp}, the graph induced by the non trivial cliques given in input is connected. This implies that the solution $S$ constructed in $\gamma$ is also connected. Consider the last property. We follow the approach proposed in~\cite{DelbotLP13}. According to Theorem 2 showed in~\cite{DelbotLP13}, $S$ is a 2-approximation for the Connected Vertex Cover problem. The approximation ratio comes from the fact that for each clique of size $k \geq 2$ at least $k-1$ nodes are in an optimal solution to cover all the $\frac{k \times (k-1)}{2}$ edges of the clique, while $k$ nodes are selected by the algorithm.
\end{proof}

\begin{theorem}
Algorithm $\mathcal{SS-CVC}$ is a self-stabilizing algorithm for Specification~\ref{spec:cvc} under an unfair distributed daemon.
\end{theorem}

\begin{proof}
We have to show that starting from any configuration the execution of Algorithm $\mathcal{SS-CVC}$ verifies the two conditions of Specification~\ref{spec:cvc}.

According to Theorem~\ref{thm:cvc_enable_action}, Lemmas~\ref{lem:config_legitime_cvc_round},~\ref{lem:config_legitime_cvc_step} and~\ref{lem:no_enable_action_cvc}, from any configuration Algorithm $\mathcal{SS-CVC}$ reaches a configuration $\gamma \in \mathcal{C}$ in finite time and $\gamma$ is a terminal configuration, which verifies Condition 1 of Specification~\ref{spec:cvc}. Moreover, according to Lemma~\ref{lem:sol_cvc_config_legitime} the terminal configuration $\gamma$ reached by Algorithm $\mathcal{SS-CVC}$ satisfies Definition~\ref{def:cvc}, which verifies Condition 2 of Specification~\ref{spec:cvc}.
\end{proof}

\section{Conclusion}
In this paper, we give the first distributed and self-stabilizing algorithm for the Connected Vertex Cover problem with a constant approximation ratio of 2. Moreover, to solve this problem we propose also a self-stabilizing algorithm for the construction of a Connected Minimal Clique partition of the graph. These two algorithms work under the unfair distributed daemon which is the weakest daemon.
There are two natural perspectives to this work. First, our distributed self-stabilizing clique partition construction a root node is used. This allows to ensure the connectivity property for the clique partition. If this property is not necessary our algorithm can be easily modified in order to remove this hypothesis, but is it also possible while guaranteeing the connectivity property. Second, the self-stabilizing algorithm we propose for the Connected Vertex Cover problem achieves a better time complexity than a self-stabilizing solution based on Savage's approach, but at the price of a higher memory complexity. So, a natural question is to investigate the existence of a distributed algorithm with a low time and memory complexity.


\bibliographystyle{alpha}
\bibliography{Biblio}

\end{document}